\documentclass{iopart}

\makeatletter
\renewcommand{\tableofcontents}{\section*{\contentsname} \@starttoc{toc}}
\makeatother
\expandafter\let\csname equation*\endcsname\reax
\expandafter\let\csname endequation*\endcsname\relax

\usepackage{bbm, amsmath, amssymb, amsthm, bm,textcomp, enumitem}
\usepackage[utf8]{inputenc}
\usepackage[german,dutch,english]{babel}
\usepackage{graphicx}

  \theoremstyle{definition}
  \newtheorem{WoDef}{Working definition}
  \newtheorem{Def}{Definition}
  \newtheorem{Exp}{Example}
\theoremstyle{theorem}
\newtheorem{Pro}{Proposition}

\setlist[enumerate]{leftmargin=*,label=(\roman*)}

\begin{document}
  \paper{Measures of macroscopicity for quantum spin systems}

  \author{Florian Fr\"owis and Wolfgang D\"ur}
  \address{Institut f\"ur Theoretische Physik, Universit\"at
    Innsbruck, Technikerstr. 25, A-6020 Innsbruck,  Austria}
  \ead{florian.froewis@uibk.ac.at}
\date{\today}

  \begin{abstract}
    We investigate the notion of ``macroscopicity'' in the case of quantum spin systems and provide two main results. First, we motivate the quantum Fisher information as a measure for the macroscopicity of quantum states. Second, we compare the existing literature of this topic. We report on a hierarchy among the measures and we conclude that one should carefully distinguish between ``macroscopic quantum states'' and ``macroscopic superpositions'', which is a strict subclass of the former.
  \end{abstract}

 \tableofcontents

 \maketitle

  \section{Introduction}
  \label{sec:introduction}

  Quantum mechanics is maybe the most successful and fascinating physical theory of the last century. In the first place, it provides a deep understanding of atoms and their interaction with light. In recent years, the promise for improved technologies is intensively investigated. This success comes with the price of a difficult interpretation of the theory, which was lively discussed from a philosophical point of view. As long as we consider only microscopic systems on the scale of an atomic radius, objections to quantum mechanics are nevertheless rare, ultimately because of the overwhelming experimental evidence. When it comes to macroscopic systems, many things are not clear any more. Already in 1935, Schr\"odinger pointed out in his seminal paper \cite{Schroedinger35} that quantum mechanics in principle allows superpositions of macroscopic states, like a cat that is alive and dead at the same time. This gedankenexperiment is deeply connected to the so-called quantum measurement problem and was discussed by generations of physicists.
   
  Besides the interest in the foundations of quantum mechanics, the question of macroscopic quantum mechanics also has practical aspects. Proposed architectures for quantum computers are based on a large number of qubits. Computational tasks that can overcome classical algorithms may require long-range quantum correlations \cite{QComp,QComp2}. A further application, quantum metrology, uses a certain kind of multipartite entanglement among many particles for an increased sensitivity in phase estimation protocols \cite{FisherEntanglement}.

  In 1980, Leggett \cite{Leg80} gave an important impulse to the topic of macroscopic quantum mechanics. He asked for a clear definition of the phrases ``macroscopic quantum phenomenon'' and ``macroscopic superposition''. He realised that one should distinguish between quantum effects that are originated on a microscopic level from ``true'' macroscopic quantum effects. Among other examples, he highlights the specific heat of insulators. Classical statistical mechanics predicts a specific heat that is constant with respect to the temperature $T$. On the other side, the quantum mechanical Debye model correctly predicts a $T^3$ behaviour for small temperatures. Many physicists considered this as an example of a macroscopic quantum phenomenon, since this law is valid even for large insulators. Leggett argued that the phrase macroscopic in this context is not justified, because the interactions that cause this effect are on an atomic scale and thus microscopic. In the following, he demanded a distinction between classical and microscopic quantum effects on the one side, and macroscopic quantum effects on the other side. Only the latter ones allow to verify quantum mechanics (against classical theories) on a macroscopic scale, as in the example of Schr\"odinger's cat.

  Consequently, we call quantum states that are capable to induce macroscopic quantum effects ``macroscopic quantum states''. The question at issue is which properties of a many-body quantum state are appropriate for such a characterisation. It is clear that the number of particles is an important but not sufficient criterion. If we consider superpositions of semi-classical quantum states, it seems to be crucial that they are ``macroscopically distinct'', as Leggett \cite{Leg80} phrased it. However, a straight forward mathematical formulation of this intuitive characterisation does not exist. Furthermore, there may be quantum states that do not exhibit a superposition of two semi-classical states, but are superpositions of a large number of those. These and further concerns led to various proposals for macroscopic quantum states \cite{Leg80,DSC,p-index,BM,q-index,KWDC,MAvD,LJ}.

  The goal of this work is twofold. First, we motivate and introduce another aspect of macroscopic quantum effects in discrete systems. After basic considerations in section \ref{sec:sett-basis-cons}, we propose to use the so-called quantum Fisher information as a measure for ``macroscopicity''. The quantum Fisher information originally appeared in the context of phase estimation \cite{ClassicalFisher}, but also gives insight into the geometry of the space of density operators \cite{DistinguishMetric} and their entanglement properties \cite{FisherEntanglement,HyllusToth,HyllusToth2}. In section \ref{sec:fish-inform-as}, we argue that the quantum Fisher information is an appropriate  measure to distinguish macroscopic from microscopic quantum effects in discrete quantum systems. Furthermore, it is a key property to judge on the usability of quantum states to overcome classical limits in parameter estimation protocols. In addition, we present a further characterisation for superpositions of the form
  \begin{equation}
    \label{eq:1}
    \left| \psi \right\rangle  = \frac{1}{\sqrt{2}} \left( \left| \psi_0 \right\rangle + \left| \psi_1 \right\rangle  \right),
  \end{equation}
  where $\left| \psi_0 \right\rangle $ and $\left| \psi_1 \right\rangle $ are orthogonal quantum states. This proposal is called ``relative Fisher information'' and compares how macroscopic $\left| \psi \right\rangle $ is in relation to the macroscopicity of $\left| \psi_0 \right\rangle $ and $\left| \psi_1 \right\rangle $. It is necessary for the latter to be semi-classical quantum states in order to form a macroscopic superposition due to the relative Fisher information. The basic motivation for this proposal is to identify the class of macroscopic superpositions in the spirit of Schr\"odinger's cat. In contrast to other proposals, it is unambiguously defined for all superpositions $\left| \psi \right\rangle $.

  Second, we compare the quantum Fisher information and the relative Fisher information as macro-measures with other measures for qubit systems that were recently proposed \cite{p-index,BM,q-index,KWDC,MAvD}. Therefore we give a short review on these measures in section \ref{sec:hier-among-macro}. There exist two classes of measures. The first class considers arbitrary quantum states, while the second class focuses on superposition states as in equation (\ref{eq:1}). In section \ref{sec:relat-among-meas}, we identify a hierarchy among those proposals. We show that the set of macroscopic superpositions (\ref{eq:1}) --classified by \cite{BM,KWDC,MAvD} and the relative Fisher information-- is a strict subset of general macroscopic quantum states according to \cite{p-index} and the quantum Fisher information.

  We summarise and conclude in section \ref{sec:conclusion}.

  \section{Setting and basic considerations}
  \label{sec:sett-basis-cons}

We exclusively focus on $N$ discrete two-level systems (qubits), defined on the Hilbert space $\mathbbm{C}^{2\otimes N}$. We restrict our discussion on qubit systems, because, on the one hand, they are valid representations of a large class of physical systems. On the other hand, the mathematical treatment of problems for qubit systems is often much simpler than for other frameworks, which can lead to insight that is useful for the entire quantum theory. However, it is clear that qubit systems are \textit{a priori} an abstract concept that does not include spatial distances between particles nor masses or energy scales. Therefore, the only parameter that gives rise to a qualification of macroscopic quantum states is the particle number $N$. In general, we assume $N$ to be a ``large number'', that is, $N \gg 1$. The boundary at which we call $N$ large is not strict. For theoretical considerations as in this paper, it is often more convenient to study the scaling of a certain property like the macroscopicity of quantum states. Therefore, we focus on quantum states that are scalable (i.e., there exists a recipe how to define these states for any $N$), and ask whether these states are macroscopic in the limit of large, but finite $N$.

As mentioned by other authors \cite{BM,KWDC,Leggett02}, the topic of macroscopic quantum mechanics is ultimately subjective. This is not only because different authors consider different properties of a quantum system as essential for the characterisation of a macro-system. (However, we will see in section \ref{sec:relat-among-meas} that different proposals lead to similar measures.) The notion of locality plays an essential role in all works that we are going to discuss in the course of this paper. This means that quantum states are macroscopic with respect to a given concept of locality. The motivation is physical rather than mathematical. It comes from the observation that all interactions we encounter in nature are of finite range. For macroscopic qubit systems, we therefore assume that all Hamiltonians and all measurements are sums of local terms. Since in this paper, we are concerned with the scaling of certain properties, we demand that the range of the addends is independent of the system size (i.e., $O(1)$ \footnote{The definitions we use here are $f(N) = O(N^x) :\Leftrightarrow \lim_{N\rightarrow \infty}f(N)/N^x > 0$ and  $f(N) = o(N^x) :\Leftrightarrow \lim_{N\rightarrow \infty}f(N)/N^x = 0$.}). For simplicity, we only consider local operators that act nontrivially on distinct groups of qubits.

  \section{Quantum Fisher information as a measure for macroscopicity}
  \label{sec:fish-inform-as}

  In this section, we define the quantum Fisher information as a measure for macroscopic quantum states. First, we begin with working definitions to motivate our proposal. These starting points are more or less vague characterisations of macroscopic quantum states, since they are built on vaguely defined terms. From these, we try to infer a more stringent definition, which we formulate in section \ref{sec:math-defint-macro}.

  \subsection{Starting points for a macro-measure}
  \label{sec:basic-idea-macro}

  As already mentioned, there are two important aspects of macroscopic quantum mechanics.  Originally, physicists have been interested in the foundations of quantum mechanics. Does a macroscopic object like a cat obey the Schr\"odinger equation \textit{in principle}? If not, where is the border between classical and quantum mechanical theory and what kind of theory could unify both realms? Those questions led, for example, to collapse theories like \cite{GRW}, which state that the Schr\"odinger equation has to be altered in order to describe the physics of heavy objects correctly. A different view-point is given by the decoherence theory. It assumes the general validity of quantum mechanics on macroscopic scales. Generically, a quantum system interacts with its environment locally, which destroys non-local quantum correlation within the system. Larger objects interact with the environment more intensively and therefore loose these correlations more rapidly. In decoherence theory, this is the reason why we never experience macroscopic quantum effects in our daily life. There is a way to provide experimental evidence in favour of the decoherence theory. Suppose we study a collapse model that forbids quantum mechanics on a certain mass scale. If we observe a behaviour of a macroscopic system at this scale that cannot be explained by classical theories nor by accumulative microscopic quantum effects (like the Debye-model for the specific heat), this collapse model is falsified. In the spirit of Leggett \cite{Leg80,Leggett02}, we take this as a starting point and note
  \begin{WoDef}\label{def1}
    A quantum system is called macroscopic if it is capable to show a behaviour that is neither a classical nor an accumulated microscopic quantum effect.
  \end{WoDef}

  In addition to fundamental considerations, applications of quantum mechanics are at the focus of interest. As a first starting point for applications analogous to working definition \ref{def1}, one could demand that a macroscopic quantum state is able to perform a certain task faster or better than any classical device. However, this leads to ambiguities for two reasons. First, there is a considerable range of different applications that are incommensurable. By focusing on a specific task to characterise macroscopic quantum states, we certainly bias the set of macro-states. If we look at resource states for universal quantum computation in a measurement based setting \cite{MBQC,MBQC2} we encounter different states than quantum states that are useful for quantum metrology. Second, the phrase ``faster or better'' is vaguely defined. In quantum computation, for example, it is not at all clear how much better certain quantum algorithms are.

  In this paper, we focus on quantum metrology for two reasons. First, the system size $N$ is a crucial parameter for the performance of a parameter estimation protocol. In addition, the improvement with the help of quantum mechanics compared to classical strategies is mathematically established. In the next section, we shortly review the theory of quantum metrology. For now, we note 

  \begin{WoDef}
    \label{def2}
    A macroscopic quantum state is defined as a resource state that is capable to increase the sensitivity of a parameter estimation protocol qualitatively.
  \end{WoDef}
  
  Before we go on with the implications of working definitions \ref{def1} and \ref{def2}, we introduce the concept of an ``effective size'' $N_{\mathrm{eff}}$ of the system. This term has been used in the literature, e.g.~in \cite{DSC,BM,KWDC,MAvD}, to assign a number $N_{\mathrm{eff}} \leq N$ to a system to judge on the ``macroscopicity'' of the system. We try to make this idea more precise with
  
  \begin{WoDef}
    \label{def3}
    A given quantum state $\rho$ of an $N$ qubit system may be capable to show a non-classical phenomenon. The effective size of $\rho$, denoted by $N_{\mathrm{eff}}(\rho)$, is the minimal system size for which one has to assume validity of quantum mechanics in order to explain this phenomenon.
  \end{WoDef}
  
The effective size gives us therefore a quantification of working definitions \ref{def1} and \ref{def2}. Given a certain quantum state $\rho$ of $N$ qubits, its effective size $N_{\mathrm{eff}}(\rho)$ is taken as the basis for a qualification of $\rho$ as a ``macroscopic quantum state''. A nontrivial step is now to define for which $N_{\mathrm{eff}}(\rho)$ one should call $\rho$ macroscopic. While for a specific experiment the absolute number of $N_{\mathrm{eff}}(\rho)$ may be crucial, for theoretical investigations the scaling of $N_{\mathrm{eff}}(\rho)$ with $N$ is apparently more interesting. The maximal scaling of the effective size is linear in $N$ due to working definition \ref{def3}. It is therefore clear that $\rho$ is called macroscopic if $N_{\mathrm{eff}}(\rho) = O(N)$. For a scaling that is sub-linear [i.e., $N_{\mathrm{eff}}(\rho) = o(N)$], the decision whether $\rho$ is called macroscopic or not is subjective and a straightforward answer does not exist. In this paper, we follow a conservative definition and use for all considered macroscopicity measures

  \begin{Def}
    A quantum state $\rho$ is called macroscopic (due to a given measure) if its effective size $N_{\mathrm{eff}}(\rho)$ (due to this measure) is linear in the system size, that is, $N_{\mathrm{eff}}(\rho)=O(N)$.\label{def:macro}
  \end{Def}

  As a counterexample, consider a product state of maximally entangled two-qubit states $\left| \psi^{-} \right\rangle ^{\otimes N/2}$, which is used to model Cooper pairs in superconductivity theory. The effective size is $N_{\mathrm{eff}} = 2$ and the phenomenon of superconductivity is hence a microscopic quantum effect according to working definition \ref{def3} and definition \ref{def:macro}.

  In the following paragraphs, we try to infer from the working definitions  mathematical definitions for macroscopic quantum states and effective size.

  \subsection{Mathematical definition for macro-measure}
  \label{sec:math-defint-macro}

  We are going to propose the use of the quantum Fisher information as a measure for macroscopic quantum states. Since in this paper, we exclusively focus on the quantum Fisher information and are not concerned with the (classical) Fisher information, the term Fisher information always refers to quantum Fisher information.
  
  Consider an initial quantum state $\rho$ that is subject to a unitary time evolution generated by the time-independent Hamiltonian $H$, $\rho(t) = e^{-iHt/\hbar} \rho e^{iHt/\hbar}$. The set $\left\{ \rho(t): t\in \mathbbm{R} \right\}$ is a parametrised curve through the space of density operators. If we use the Bures metric \cite{Bures} to measure distances in this space, we define the quantum Fisher information implicitly as \cite{DistinguishMetric,Uhlmann92}
\begin{equation}
\label{eq:5}
(ds)_{\mathrm{Bures}} = \frac{1}{2\hbar}\sqrt{\mathcal{F}_t(\rho,H)} dt.
\end{equation}
In general, $\mathcal{F}_t(\rho,H)$ depends on $t$. In our case, $H$ is time-independent and so is $\mathcal{F}_t(\rho,H)$ \cite{DistinguishMetric}, for which reason we omit the index $t$ in the following. Then, the explicit formula of $\mathcal{F}(\rho,H)$ reads \cite{DistinguishMetric,HelstromHolevo,HelstromHolevo2}
\begin{equation}
\label{eq:6}
\mathcal{F}(\rho,H) = 2\sum_{i,j=1}^{2^N}\frac{\left( \pi_i-\pi_j \right)^2}{\pi_i+\pi_j}\left| \left\langle i \right| H \left| j \right\rangle  \right|^2,
\end{equation}
where we used the spectral decomposition of $\rho$
\begin{equation}
\label{eq:7}
\rho = \sum_{i=1}^{2^N} \pi_i \left| i \rangle\!\langle i\right| .
\end{equation}
For pure states $\rho=\left| \psi \rangle\!\langle \psi\right| $, the Fisher information reduces to the variance $ \mathcal{V}_{\psi}(H) = \langle H^2 \rangle_{\psi}-\langle H \rangle_{\psi}^2$ of $H$
\begin{equation}
\label{eq:9}
\mathcal{F}(\psi,H) = 4 \mathcal{V}_{\psi}(H).
\end{equation}

As discussed in the introduction, in physical systems of macroscopic size, we generically encounter Hamiltonians that consists of local terms. These describe interactions with external fields like a magnetic field or few-body interactions. It has been noticed that entangled states can exhibit a much larger Fisher information than separable states. Let us focus on $H$ as a sum of one-particle terms $H = \sum_{i=1}^N H^{(i)}$, where $H^{(i)}$ acts non-trivially only on qubit $i$. We fix the operator norm of the single terms to a characteristic energy $\lVert H^{(i)} \rVert = e_0 $. It is by now well known that separable states $\rho_{\mathrm{sep}}$ obey for every local $H$ a Fisher information that scales at most linearly with the system size \cite{FisherEntanglement}, $\mathcal{F}\left(\rho_{\mathrm{sep}},H\right) \leq 4 e_0^2 N$. On the other side, there exist quantum states like
\begin{Exp}[GHZ state]
  \label{ex:ghz}
The multipartite GHZ state
\begin{equation}
\label{eq:8}
\left| \mathrm{GHZ} \right\rangle = \frac{1}{\sqrt{2}}\left( \left| 0 \right\rangle ^{\otimes N} + \left| 1 \right\rangle ^{\otimes N} \right)
\end{equation} gives $\mathcal{F}\left( \mathrm{GHZ}  ,e_0\sum_i\sigma_z^{(i)}\right) = 4 e_0^2 N^2$. (The state vectors $\left| 0 \right\rangle $ and $\left| 1 \right\rangle $ denote the eigenstates of the $\sigma_z$ Pauli operator.)   
\end{Exp}

We therefore observe that for local Hamiltonians certain quantum states exhibit a much larger Fisher information than any separable state. We use this insight for 
\begin{Def}[\textbf{a}]
  \label{defApplication}
  Be $\mathcal{A}$ the set of linear operators $A$ that are sums of one-particle terms $A = \sum_{i=1}^NA^{(i)}$, where each term has a unit operator norm, $\lVert A^{(i)} \rVert=1$. We call a quantum state $\rho$ macroscopic if there exists an $A \in \mathcal{A}$ such that $\mathcal{F}(\rho,A) = O(N^2)$. The effective size is defined as
\begin{equation}
N_{\mathrm{eff}}^{\mathrm{F}} (\rho)= \max_{A \in \mathcal{A}}\mathcal{F}(\rho,A)/(4 N).\label{eq:13}
\end{equation}
We call $\rho$ a macroscopic quantum state, if $N_{\mathrm{eff}}^{\mathrm{F}} (\rho)=O(N)$.
The superscript $\mathrm{F}$ indicates the Fisher information and ease the comparison with other proposals in sections \ref{sec:relat-fish-inform-1} and \ref{sec:hier-among-macro}.
\end{Def}

For all pure states $\left| \psi \right\rangle $, the possible range of the effective size is $1\leq N_{\mathrm{eff}}^{\mathrm{F}}(\psi)\leq N$. The upper bound comes from the operator norm of $A \in \mathcal{A}$, which equals $N$. The lower bound is a consequence of the fact that one can always find a  $A \in \mathcal{A}$ for which $\mathcal{V}_{\psi}(A) \geq N$. For general quantum states, one has $0\leq N_{\mathrm{eff}}^{\mathrm{F}}(\rho)\leq N$.
We now discuss why definition \ref{defApplication} (a) is a consequence of working definitions \ref{def1}, \ref{def2} and \ref{def3}. 

\textit{Re} working definition \ref{def1}: Suppose one measures the time-evolving state $\rho(t)$ by projecting it onto the range of the initial state $\rho(0)$. We denote this projective measurement by $\Pi$. The expectation value $\langle \Pi \rangle_{\rho(t)}$ is the probability that, at time $t$, we find the quantum state in the subspace that is the range of $\rho(0)$, for which reason this quantity is often called survival probability. Trivially, one has $\langle \Pi \rangle_{\rho(0)}=1$. In \cite{Fleming}, the following statement has been proven.
\begin{Pro}\label{Fleming}
  For the interval $0 \leq \sqrt{\mathcal{F}(\rho,H)}
  \left| t \right| /\hbar\leq \pi$ we find
  \begin{equation}
    \label{eq:11}
    \langle \Pi \rangle_{\rho(t)}\geq \cos^2 \frac{\sqrt{\mathcal{F}(\rho,H)}t}{2\hbar}.
  \end{equation}
\end{Pro}
This bound gives an upper limit on the ``speed'' with which the system evolves. Observe that inequality (\ref{eq:11}) allows certain quantum states, like the GHZ state, to evolve much faster than any separable state. This means that the observation of a fast time evolution indicates the validity of quantum mechanics. This effect is a true macroscopic quantum phenomenon, because we need long-range quantum correlations within the state $\rho(t)$ to achieve a Fisher information that is proportional to $N^2$ (compare also with \cite{Morimae09}). This aspect is discussed later in more detail.

We stress that the observation of rapid oscillations apparently verifies macroscopic quantum effects only if we have full knowledge of the unitary dynamics. If the interaction strength of the Hamiltonian is not known, the plain data does not give any conclusion on the ``macroscopicity'' of initial quantum state. However, if in the experiment we are able to vary $N$, we can reveal the dependency of the measured frequency of the oscillation. A linear dependence in $N$ is an evidence for a macroscopic quantum effect.

\textit{Re} working definition \ref{def2}: The Fisher information appears most prominently in the context of metrology, which we are going to review now. In this theory, an experimenter aims to estimate an unknown parameter $\omega$ that is encoded into a quantum state she or he has access to. The state is described by a density operator $\rho(\omega)$. We measure the system with a POVM represented by the measurement operators $\left\{ E_i \right\}_i$. The probabilities $p_i(\omega)= \mathrm{Tr}\left[ \rho(\omega)E_i \right]$ of the outcomes give rise to an estimate of $\omega$. The general difficulty of calculating the error of the estimate $\delta\omega$ is eased by the Cram\'er-Rao bound \cite{ClassicalFisher,CramerRao}
\begin{equation}
\label{eq:2}
\delta\omega \geq \frac{1}{\sqrt{n F(\omega)}}
\end{equation}
where $n$ is the total number of repetitions of the experiment and $F(\omega)$ is the classical Fisher information \cite{ClassicalFisher}
\begin{equation}
\label{eq:3}
F(\omega) = \sum_i p_i(\omega)\left[ \frac{d}{d\omega}\log p_i(\omega) \right]^2.
\end{equation}
The higher the classical Fisher information, the better is the lower bound on $\delta\omega$, which can be attained for unbiased estimates in the limit of large $n$. One can improve the Cram\'er-Rao bound by varying the measurement one uses. Again, we then end up with the quantum Fisher information, which within this approach can be defined as \cite{DistinguishMetric}
\begin{equation}
\label{eq:4}
\mathcal{F}(\omega) = \max_{\left\{ E_i \right\}_i}F(\omega).
\end{equation}

If the curve $\rho(\omega)$ is generated by an $\omega$-independent self-adjoint operator $A$, we end up with the same formula for the Fisher information $\mathcal{F}(\rho,A)$ from equation (\ref{eq:6}). The previous discussion on the implications for the measure of macroscopic quantum states is similar for quantum metrology. While separable states lead to a minimal error $\delta\omega \propto 1/\sqrt{N}$, ``macroscopic'' quantum states as a resource give rise to a quadratic improvement $\delta\omega \propto 1/N$. From working definition \ref{def2}, we again find that a quantum state is macroscopic if $\mathcal{F}(\rho,H) = O(N^2)$, see also definition \ref{defApplication} (a).

\textit{Re} working definition \ref{def3}: The assignment $N_{\mathrm{eff}}^{\mathrm{F}}(\rho) = \max_{A \in \mathcal{A}} \mathcal{F}(\rho,A)/(4N)$ was defined as the effective size of $\rho$. We now argue why this quantity allows us to distinguish between microscopic and macroscopic quantum effects in spin systems. We therefore divide the Hilbert space into groups of at most $k$ particles $\mathcal{H} =  \mathbbm{C}^{2\otimes k_1} \otimes \ldots \otimes \mathbbm{C}^{2 \otimes k_n}$, with $\sum_{i = 1}^n k_i = N$ and $k_i \leq k$. A pure state $\left| \psi_k \right\rangle = \left| \phi_{1} \right\rangle \otimes \ldots \otimes \left| \phi_n \right\rangle$ with $\left| \phi_{j} \right\rangle \in \mathbbm{C}^{2\otimes k_j}$ that is a tensor product with respect to this splitting is called $k$-producible. For local $H$, the variance is additive
\begin{equation}
\label{eq:47}
\mathcal{V}_{\psi_k}(H) = \mathcal{V}_{\phi_1}\left(\sum_{i=1}^{k_1}H^{(i)}\right) + \ldots +\mathcal{V}_{\phi_n}\left(\sum_{i=N-k_n}^{N}H^{(i)}\right).
\end{equation}
For every group $j$ of equation (\ref{eq:47}), the variance can be estimated by the maximal possible value $k_j^2 \leq k^2$. Then, one can easily see that the Fisher information of $\left| \psi_k \right\rangle $ is at most 
\begin{equation}
\mathcal{F}\left(\psi_k,\frac{H}{2 e_0}\right) \leq k_1^2 + \ldots + k_n^2 \leq k \left( k_1 + \dots + k_n \right) = kN.\label{eq:15}
\end{equation}
A Fisher information $\mathcal{F}(\psi_k,H) \propto kN$ for the given splitting $\left\{ k_i \right\}$ is possible only if quantum mechanics is valid (compared to classical physics) on the scale of $k$ particles. 

Very recently, a tighter and generalised estimate was shown for $k$-producible mixed states $\rho_k$, which are defined as incoherent mixtures of $k$-producible pure states with respect to possibly distinct splittings $\left\{ k_i \right\}$. The authors of \cite{HyllusToth,HyllusToth2} proved that
\begin{equation}
\label{eq:10}
\mathcal{F}\left(  \rho_k,\frac{H}{2e_0}\right) \leq  s k^2 + (N-sk)^2
\end{equation}
with $s = \lfloor \frac{N}{k}\rfloor$, which coincides with the result for pure states if $k$ is a factor of $N$. 

Therefore, if a quantum state $\rho$ has an effective size $N_{\mathrm{eff}}^{\mathrm{F}} \geq k$, we know that this state is potentially able to demonstrate the existence of quantum correlations of the range of at least $k$ qubits.

\subsection{Discussion and basic examples}
\label{sec:discussion}

In this section, we discuss some implications and aspects of definition \ref{defApplication} (a) and illustrate them with examples.

The first paradigmatic macroscopic quantum state we discuss is the GHZ state (\ref{eq:8}). As already mentioned, the maximal variance of a dimensionless local operator equals $N^2$, which means that the effective size (\ref{eq:13}) [with (\ref{eq:9})] is $N_{\mathrm{eff}}^{\mathrm{F}}( \mathrm{GHZ}  ) = N$. This is in accordance to all other proposals \cite{Leg80,DSC,p-index,BM,KWDC,MAvD,LJ}.

Prominent quantum states that do not exhibit long-range quantum correlations are the W state \cite{WState} and the cluster states \cite{BR01,Hein05}. Therefore, these states are not macroscopic with respect to the Fisher information. Observe, however, that the Fisher information is not an entanglement measure as it can increase under local operations and classical communication (LOCC). For example, from the two-dimensional cluster state with $N$ qubits, one can obtain a GHZ state with $O(N)$ qubits by means of LOCC \cite{BR01}. Because the Fisher information detects a certain kind of correlations, the Fisher information can increase under LOCC. We conclude that if we allow LOCC to manipulate quantum states before we calculate the Fisher information, we enlarge the set of macroscopic quantum states, for example, by the two-dimensional cluster state, which is a valuable resource, as it can be used for universal measurement based quantum computation \cite{MBQC,MBQC2}. In this paper, we calculate the Fisher information without previous LOCC manipulation, which is in line with other discussed proposals \cite{Leg80,p-index,BM,q-index,KWDC,MAvD,LJ}. Note, however, that in \cite{DSC}, quantum states were called macroscopic, if one can distill the GHZ state with $O(N)$ particles from the resource state. This clearly uses LOCC and in this sense, also the cluster states are macroscopic.

Several works on macroscopic quantum states \cite{BM,KWDC,MAvD,LJ} have discussed a specific quantum state, which was introduced in \cite{DSC}. It is an example of a superposition, where the two constituents $\left| \psi_0 \right\rangle $ and $\left| \psi_1 \right\rangle $ are not orthogonal but exhibit an overlap that vanishes in the limit of large $N$.

\begin{Exp}[Generalized GHZ state]\label{ex:generalized-ghz}
  This state, which we call ``generalised GHZ'' state in the following, is
  defined as
  \begin{equation}
    \label{eq:12}
    \left| \phi_{\epsilon} \right\rangle = \frac{1}{\sqrt{2 \left( 1+  \cos^N \epsilon \right)}}\left( \left| 0 \right\rangle ^{\otimes N} + \left| \epsilon \right\rangle ^{\otimes N} \right)
  \end{equation}
  with $\left| \epsilon \right\rangle = \cos \epsilon \left| 0
  \right\rangle + \sin \epsilon \left| 1 \right\rangle $, $\epsilon \in [0,\pi/2]$. The authors
  of \cite{DSC,BM,KWDC,MAvD,LJ} conclude unanimously that the
  effective size $N_{\mathrm{eff}}$ (in the sense of the respective
  contributions) of the quantum state (\ref{eq:12}) for $\epsilon \ll
  1$ equals $N_{\mathrm{eff}} \approx \epsilon^2 N$, that is the
  effective size of the generalised GHZ state compared to the standard
  GHZ state is reduced by the factor $\epsilon^2$ (see also table \ref{tab:summaryExamples}). To quantify the
  effective size via the Fisher information, we have to maximise over
  all local observables $A$ (see also reference \cite{Hyllus}). Due to the symmetry of the quantum state,
  we optimise the variance $\mathcal{V}_{\phi_{\epsilon}}(A)$ with the
  ansatz
  \begin{equation}
    \label{eq:14}
    A = \sum_{i=1}^N \cos \alpha \sigma_x^{(i)} + \sin \alpha \sigma_z^{(i)}.
  \end{equation}
  A simple maximisation over the angle $\alpha$ leads to $\alpha
  = \epsilon/2$ and hence $\mathcal{V}_{\phi_{\epsilon}}(A_{\mathrm{max}}) =
N^2 \frac{\sin^2\epsilon}{1-\cos^N\epsilon} + N(1-\frac{\sin^2\epsilon}{1-\cos^N\epsilon})$. With $\cos^N \epsilon  \approx 0$ for large $N$ one has $N_{\mathrm{eff}}^{\mathrm{F}}(\phi_{\epsilon}) \approx N\sin^2
  \epsilon + \cos^2 \epsilon \approx \epsilon^2 N$. The last approximation
  is valid for small $\epsilon$ and it is therefore in accordance with
  \cite{DSC,BM,KWDC,MAvD,LJ}.
\end{Exp}

 The $\left| \phi_{\epsilon} \right\rangle $ state is a possible generalisation of the standard GHZ state. An other possible extension is to superpose slightly entangled states instead of product states. A state of this kind was discussed in \cite{ClusterGHZ,cGHZpra}.
Consider the one-dimensional cluster states with periodic boundary
  conditions
  \begin{equation}
    \label{eq:16}  
      \left| \mathrm{Cl^{+}} \right\rangle = \prod_{i=1}^N
      C^{(i,i+1)} \left|
        + \right\rangle ^{\otimes N} \quad \text{and} \quad
      \left| \mathrm{Cl^{-}} \right\rangle = \prod_{i=1}^N
      C^{(i,i+1)} \left| - \right\rangle ^{\otimes N},
  \end{equation}
  where $C^{(i,i+1)} = \left| 0 \rangle\!\langle 0\right|
  ^{(i)}\mathbbm{1}^{(i+1)} + \left| 1 \rangle\!\langle 1\right|
  ^{(i)} \sigma_z^{(i+1)}$ is a phase gate, $\left|\pm \right\rangle $
  are the $\pm 1$ eigenstates of the $\sigma_x$ Pauli operator and $N+1
  \equiv 1$.

  \begin{Exp}[Cluster-GHZ state]\label{ex:cluster-ghz}
  The superposition
  \begin{equation}
    \label{eq:17}
    \left| \mathrm{CG} \right\rangle = \frac{1}{\sqrt{2}}\left( \left| \mathrm{Cl^{+}} \right\rangle + \left| \mathrm{Cl^{-}} \right\rangle  \right)
  \end{equation}
  was called cluster-GHZ state in \cite{cGHZpra}. One can see easily
  that the cluster states (\ref{eq:16}) do not have long
  range quantum correlations and neither does the cluster-GHZ state. According to definition \ref{defApplication} (a), one has $N_{\mathrm{eff}}^{\mathrm{F}} = 1$. On the other side, we have long-range quantum correlations among local \textit{groups} of particles. If we consider two groups of three particles each, we see with
  $S_i=\sigma_z^{(i-1)}\sigma_x^{(i)}\sigma_z^{(i-1)}$ that for all
  $(i,j)$
  \begin{equation}
    \label{eq:18}
    \langle S_iS_j \rangle_{ \mathrm{CG}  }-\langle S_i \rangle_{ \mathrm{CG}  }\langle S_j \rangle_{ \mathrm{CG}  }=1,
  \end{equation}
  since $S_i \left| \mathrm{Cl^{\pm}} \right\rangle = \pm\left|
    \mathrm{Cl^{\pm}} \right\rangle $. Hence, one has $\mathcal{V}_{\mathrm{CG}}(\sum_{i=2,5,8,\ldots}S_i) = N^2/3$ and we would intuitively assign $N_{\mathrm{eff}}^{\mathrm{F}}(\mathrm{CG}) = N/3$.
\end{Exp}

Also other states like \textit{logical} GHZ states, where the physical particles in equation (\ref{eq:8}) are replaced by blocks of qubits \cite{cGHZpra}, may not have long-range quantum correlations between single particles, but between groups of size $O(1)$. Since it seems natural to consider these states also as macroscopic, we have to extend definition \ref{defApplication} (a) by
\addtocounter{Def}{-1} 
\begin{Def}[\textbf{b}]\label{DefExtension}
  Let the Hilbert space  $\mathbbm{C}^{2\otimes N}$ be divided into $n = O(N)$ distinct groups of qubits, each group is of size $O(1)$. Let $\mathcal{A}$ be the set of all local operators $A=\sum_{i=1}^{n} A^{(i)}$ such that every operator $A^{(i)}$ acts non-trivially on the group $i$ and $\lVert A^{(i)} \rVert = 1$. A quantum state $\rho$ is called macroscopic if there exists a grouping such that
  \begin{equation}
\label{eq:19}
N_{\mathrm{eff}}^{\mathrm{F}}(\rho)= \max_{A\in \mathcal{A}} \mathcal{F}(\rho,A)/(4n) = O(N).
\end{equation}
\end{Def}
In the following, the term ``macroscopic quantum state'' refers to the definitions \ref{defApplication} (a) and \ref{DefExtension} (b), together denoted by definition \ref{DefExtension}. Although it should be clear from the definition, this measure for macro-states is valid for pure and mixed states. As we see later in section \ref{sec:hier-among-macro}, for other measures this is not the case, especially if they are defined for superposition states like in equation (\ref{eq:1}).

For a fixed grouping with $n$ groups in total, the maximal possible effective size is $N_{\mathrm{eff}}^{\mathrm{F}} = n$. It is therefore clear that in general we try to maximise $n$.

\subsection{Relative Fisher information for macroscopic superpositions}
\label{sec:relat-fish-inform-1}

The Fisher information proved to be a good candidate for measuring the effective size of a macroscopic quantum states. It is well defined for general mixed states and does not require specific superpositions like in equation (\ref{eq:1}), as other measures \cite{BM,KWDC,MAvD} do. The famous Schr\"odinger cat gedankenexperiment, however, is an example where two classical states are superposed. The total state is \textit{the} archetypal macroscopic superposition. In order to have a characterisation of the ``catness'' of $\left| \psi \right\rangle = \frac{1}{\sqrt{2}}\left( \left| \psi_0 \right\rangle +\left| \psi_1 \right\rangle  \right)$ at hand, we introduce a measure for superpositions that is based on the Fisher information.

\begin{Def}
  \label{def:relative-Fisher}
  Be $\left| \psi \right\rangle = \frac{1}{\sqrt{2}}\left( \left| \psi_0 \right\rangle +\left| \psi_1 \right\rangle  \right) \in \mathbbm{C}^{2 \otimes N}$ a superposition of two quantum states $\left| \psi_0 \right\rangle $ and $\left| \psi_1 \right\rangle$. The quantity
\begin{equation}
\label{eq:46}
N_{\mathrm{eff}}^{\mathrm{rF}}(\psi) = \frac{N_{\mathrm{eff}}^{\mathrm{F}}(\psi)}{\frac{1}{2}N_{\mathrm{eff}}^{\mathrm{F}}(\psi_0)+\frac{1}{2}N_{\mathrm{eff}}^{\mathrm{F}}(\psi_1)}
\end{equation}
is called relative Fisher information. If $N_{\mathrm{eff}}^{\mathrm{rF}} = O(N)$, then $\left| \psi \right\rangle $ is called a macroscopic superposition.
\end{Def}
The measure is such that a superposition state is macroscopic, only if the Fisher information of the total state is high and at the same time the constituents $\left| \psi_0 \right\rangle $ and $\left| \psi_1 \right\rangle$ have a small Fisher information, which means that they considered to be semi-classical, that is, microscopic.

There are two important differences between $N_{\mathrm{eff}}^{\mathrm{F}}$ and $N_{\mathrm{eff}}^{\mathrm{rF}}$.

(1) There are states $\left| \psi \right\rangle $ that have a small $N_{\mathrm{eff}}^{\mathrm{rF}}(\psi)$ and are nevertheless able to show macroscopic quantum effects demanded in working definition \ref{def1}, e.g.~if $N_{\mathrm{eff}}^{\mathrm{F}}(\psi) = O(N)$ and $N_{\mathrm{eff}}^{\mathrm{F}}(\psi_0) = O(N)$ (see example \ref{Ex:Dicke} later).

(2) In contrast to the Fisher information, the extension of definition \ref{def:relative-Fisher} to mixed states is not trivially possible. However, if we consider the an initially pure state $\left| \psi \right\rangle = \frac{1}{\sqrt{2}}\left( \left| \psi_0 \right\rangle +\left| \psi_1 \right\rangle  \right)$ that is subject to a cp-map $\mathcal{E}$, the relative Fisher information for a the noisy state $\rho = \mathcal{E}(\psi)$ may be defined as
\begin{equation}
\label{eq:23}
N_{\mathrm{eff}}^{\mathrm{rF}}(\rho) = \frac{N_{\mathrm{eff}}^F[\mathcal{E}( \psi  )]}{\frac{1}{2}N_{\mathrm{eff}}^F[\mathcal{E}( \psi_0  )]+\frac{1}{2}N_{\mathrm{eff}}^F[\mathcal{E}( \psi_1 )]}.
\end{equation}

  \section{Comparison to existing proposals}
  \label{sec:hier-among-macro}

  In general, the genesis of proposals for macroscopic quantum states is the following. First, one tries to find a property of quantum states that seems to be characteristic for genuine macroscopic states. (For the Fisher information, this step is manifested in working definitions \ref{def1} to \ref{def3}.) Next, this property is defined in a mathematically rigorous way. Finally, the definition is applied to several examples to check for consistency. In this section, we review and discuss the first two points of several measures \cite{p-index,BM,KWDC,MAvD,q-index} that are suitable for multipartite qubit systems. We find that there exist two main classes of measures. The first class, which incorporates the Fisher information and \cite{p-index,q-index}, tries to quantify general quantum states of macroscopic systems. The second class --consisting of the relative Fisher information and the measures \cite{BM,KWDC,MAvD}-- focuses on superpositions of two orthogonal quantum states as in equation (\ref{eq:1}). To distinguish these two classes, we call macro-states from the first class ``macroscopic quantum states'' and those from the second class ``macroscopic superpositions''. The quantum states that are discussed as examples in the course of this section are summarised in table \ref{tab:summaryExamples}.

\label{sec:revi-curr-liter}
Before we start with the review of \cite{p-index,BM,q-index,KWDC,MAvD}, we remark that there are two other proposals \cite{Leg80,DSC} for macroscopic superpositions that were important for the discussion of this topic but are nevertheless not general enough to be applicable for general superpositions.

The first one was introduced by Leggett \cite{Leg80}. He demanded that two superposed states $\left| \psi_0 \right\rangle $ and $\left| \psi_1 \right\rangle $ should be macroscopically distinct. To translate this into mathematical terms, a superposition state $\left| \psi \right\rangle = \frac{1}{\sqrt{2}}\left( \left| \psi_0 \right\rangle +\left| \psi_1 \right\rangle  \right)$ is called macroscopic --according to Leggett-- if we need at least an $O(N)$-correlation operator in order to be able to distinguish between $\left| \psi \right\rangle $ and the incoherent mixture $\left| \psi_0 \rangle\!\langle \psi_0 \right| + \left| \psi_1 \rangle\!\langle \psi_1\right| $. For the GHZ state (see example \ref{ex:ghz}), we need the correlation operator $\sigma_x^{\otimes N}$ for this task. Unfortunately, the requirement is too insensitive, as it can be seen for the generalised GHZ state (example \ref{ex:generalized-ghz}). In this case, any $\epsilon>0$ leads to a macroscopic quantum state. The idea of ``macroscopic distinctness'' was nevertheless used in \cite{KWDC} and \cite{MAvD} and led to more elaborated measures.

As already mentioned, the generalised GHZ state was introduced in \cite{DSC} to introduce the notion of an effective size. There, it is shown that the trace norm\footnote{The definition of the trace norm for a linear operator $M$ reads $\lVert M \rVert_1 = \mathrm{Tr}\sqrt{M^{\dag}M}$.} of the coherence term $\left| 0 \rangle\!\langle \epsilon\right| ^{\otimes N}$ decays under local phase noise 
with a rate
\begin{equation}
\label{eq:42}
\lVert \left| 0 \rangle\!\langle \epsilon\right| ^{\otimes N} \rVert_1 \approx e^{-\gamma \epsilon^2 N t}
\end{equation}
in the case of small $\epsilon$ and $t$. It was argued that this indicates a reduced effective size compared to the GHZ state, whose coherences decay with $e^{-\gamma N t}$. However, this insight can not be generalised, that is, one cannot conclude from the damping rate of the coherences on the effective size. A counter example are so-called ``logical'' GHZ state. Suppose we replace every physical particle by an $O(1)$ group of qubits. On this group, one defines two orthogonal states $\left| 0_L \right\rangle $ and $\left| 1_L \right\rangle $. The logical GHZ state is then defined as
\begin{equation}
\label{eq:50}
\left| \mathrm{GHZ}_L \right\rangle = \frac{1}{\sqrt{2}}\left( \left| 0_L \right\rangle ^{\otimes N}+\left| 1_L \right\rangle ^{\otimes N} \right).
\end{equation}
The effective size of those logical GHZ states is $N$ due to the Fisher information, but also due to other measures \cite{p-index,BM,q-index,KWDC,MAvD} that are discussed in this paper. However, there exist examples like the so-called ``Concatenated GHZ state'' \cite{cGHZprl} that exhibit much more stable coherences than the standard GHZ state. Therefore, the mere decay rate of the coherences is not a suitable measure for macroscopic quantum states.

In \cite{DSC}, a second criterion for macroscopic superpositions was suggested. From the generalised GHZ state, one can distill a GHZ state of approximately $\epsilon^2N $ qubits. As already discussed in section \ref{sec:discussion}, this approach uses LOCC operations and is therefore in contrast to other proposals discussed in this paper.

Finally, the measure proposed in \cite{LJ} is defined for continuous variable systems. A direct application to qubit systems is not given and we therefore focus on the other proposals in the following.

\subsection{Indices p and q}
\label{sec:index-p-q}

In classical statistical mechanics, we normally deal with probability distributions in phase space (here we assume $N$ particles) whose uncertainties with respect to a given observable are small compared to its spectral radius. More quantitatively, the expectation values of extensive variable (like the magnetisation $M$ of a ferro-magnet) scales linearly with the system size $\mu \equiv \langle M \rangle_{\mathrm{cl}} = O(N)$, while for the standard deviation $\left( \Delta M \right)_{\mathrm{cl}} = \sqrt{\langle (M-\mu)^2 \rangle_{\mathrm{cl}}}$ we have $(\Delta M)_{\mathrm{cl}}=O(\sqrt{N})$. This is because for typical classical probability distributions in statistical mechanics the particles are uncorrelated.

In quantum mechanics, there exist pure states like the GHZ state (\ref{eq:8}) that exhibits an anomalously large standard deviation. For $M = \sum_{i=1}^N \sigma_z^{(i)}$ we find $(\Delta M)_{\mathrm{GHZ}} = N$. Based on this observation, the authors of \cite{p-index}, Shimizu and Miyadera, define a measure for macroscopic quantum systems, which they call index p: We consider a pure state $\left| \psi \right\rangle $. If there exists a local operator $A \in \mathcal{A}$ [\textit{cf.} definition \ref{defApplication} (a)] such that $\mathcal{V}_{\psi}(A) = O(N^p)$ with $p=2$, then $\left| \psi \right\rangle $ is called macroscopic. We comment on several aspects of this definition.

(1) Although the index p is motivated differently than the Fisher information, for pure states it is mathematically identical to the first part of definition \ref{defApplication} (a), since for pure states, the Fisher information is proportional to the variance, see equation (\ref{eq:9}).

(2) The index p does \textit{a priori} not give a quantification in form of an effective size. The discussion of the generalised GHZ state in section \ref{sec:discussion} nevertheless suggests also for the index p an assignment of $N_{\mathrm{eff}}$ as for the Fisher information, see equation (\ref{eq:13}) and the discussions in section \ref{sec:discussion}. We define
\begin{equation}
\label{eq:40}
N_{\mathrm{eff}}^{\mathrm{S}}(\psi) = \max_{A\in \mathcal{A}} \mathcal{V}_{\psi}(A)/N.
\end{equation}
With this definition, one has $N_{\mathrm{eff}}^{\mathrm{S}} =N_{\mathrm{eff}}^{\mathrm{F}} $.

(3) Regarding the Fisher information, the index p in its original definition is insensitive against long-range quantum correlations between blocks rather than between single particles. As already mentioned in \cite{Morimae10}, it therefore makes sense to extend the index p to sums of local operators that act non-trivially on groups of $O(1)$ qubits, as made precise in definition \ref{DefExtension} (b).

(4) The index p is not appropriate for mixed states and makes intuitively wrong statements for those. For instance, the mixture $\rho=\frac{1}{2} \left| 0 \rangle\!\langle 0\right| ^{\otimes N} + \frac{1}{2} \left| 1 \rangle\!\langle 1\right| ^{\otimes N}$ gives the same standard deviation as the GHZ state for the same observable $M$. To have a measure for mixed states at hand, in \cite{q-index} the so-called index q for general states was proposed. For pure states, the indices p and q qualify the same set of quantum states to be macroscopic. We give here a slightly modified definition, which was later published \cite{q-ind-altern}: Consider a quantum state $\rho$ and a local observable $A \in \mathcal{A}$. Calculate the trace norm of $C(\rho) = \left[ A,\left[ A,\rho \right] \right]$, $\lVert C(\rho) \rVert_1 = O(N^q)$. If there exist any $A\in \mathcal{A}$ such that $q=2$, then $\rho$ is a called a macroscopic quantum state.

Since the scaling of the Fisher information and the index q coincide for pure states, it is interesting whether they also coincide for mixed states. This question is open at the moment. By means of numerical studies we find that in general the difference of Fisher information for a given $\rho$ and the trace norm of $C(\rho)$ is indefinite, that is, the Fisher information cannot be bounded by the index q and vice versa.

\subsection{Relative improvement of interference experiments}
\label{sec:relat-fish-inform}

The following work of Bj\"ork and Mana \cite{BM} considers superpositions like the quantum state in equation (\ref{eq:1}). The authors try to give an operational interpretation for their measure. They consider a concrete experiment and a ``natural'' Hamiltonian $H$ associated with it. The main idea is that the constituting states $\left| \psi_0 \right\rangle $ and $\left| \psi_1 \right\rangle $ of equation (\ref{eq:1}) should be of semi-classical nature whereas the composed state (\ref{eq:1}) 
shows significant advantages for interferometric applications, that is, increased phase sensitivity. In mathematical terms, the state evolves under the time-independent Hamiltonian $H$, $\left|\psi(t)  \right\rangle  = \exp(-i Ht/\hbar)\left| \psi \right\rangle $.
Let us consider the minimal time $\theta(\psi)$ for which the state $\left| \psi(t) \right\rangle $ evolves to an orthogonal state. We call $\theta$ orthogonalization time. Bj\"ork and Mana demand that a macroscopic superposition $\left| \psi \right\rangle $ should exhibit a $\theta( \psi )$ that is much smaller than the orthogonalization time for the constituting states $\left| \psi_0 \right\rangle $ and $\left| \psi_{1} \right\rangle $. They therefore define\footnote{Note that the original quantity that was considered in \cite{BM} was defined as $\sqrt{N_{\mathrm{eff}}^{\mathrm{B}}}$. However, we will refer to (\ref{eq:41}) as the effective size due to Bj\"ork and Mana for an easier comparison to the other proposals.}
\begin{equation}
  \label{eq:24}
  N_{\mathrm{eff}}^{\mathrm{B}}(\psi) = \left(  \frac{\theta(  \psi_0)+\theta( \psi_1  )}{\theta( \psi  )}\right)^2.
\end{equation}
as measure for the effective size of a superposition (\ref{eq:1}). It has been shown in \cite{BM} that this expression can be approximated to ease the calculation of equation (\ref{eq:24}), which prove difficult in general 
\begin{equation}
\label{eq:41}
N_{\mathrm{eff}}^{\mathrm{B}}(\psi) \approx \left(  \frac{\left| \langle H \rangle_{\psi_0}- \langle H \rangle_{\psi_1} \right|}{\sqrt{\mathcal{V}_{\psi_0}(H)}+\sqrt{\mathcal{V}_{\psi_1}(H)}}\right)^2
\end{equation}

The expression (\ref{eq:24}) is well defined only if $\left| \psi_0 \right\rangle $ and $\left| \psi_1 \right\rangle $ exhibit finite orthogonalization times. For equation (\ref{eq:41}),  $\left| \psi_0 \right\rangle $ and $\left| \psi_1 \right\rangle $ should not be eigenstates of $H$. To circumvent this problem, we suggest to maximise the denominator and the numerator of equation (\ref{eq:41}) separately over all ``realistic'' Hamiltonians. As already discussed in section \ref{sec:sett-basis-cons}, this reduces the set of possible $H$ to sums of local terms, that is, the set $\mathcal{A}$ from definition \ref{DefExtension} (b). The maximal scaling of $N_{\mathrm{eff}}^{\mathrm{B}}$ is with $O(N)$. Therefore, we define the superposition $\left| \psi \right\rangle $ to be macroscopic, if $N_{\mathrm{eff}}^{\mathrm{B}} = O(N)$.

\subsection{Local distinguishability}
\label{sec:local-dist}

We now review the work of Korsbakken and coworkers \cite{KWDC}. Their notion of the effective size of a macroscopic superposition is defined as follows: Given the superposition of equation (\ref{eq:1}) with $\langle \psi_0 | \psi_1\rangle = o(1)$ (which means that the overlap vanishes for large $N$), we divide the particles in a maximal number $n$ of distinct groups such that we can distinguish $\left| \psi_0 \right\rangle$ from $\left| \psi_1 \right\rangle $ with ``high'' probability $P=1-\delta$ by measuring only \textit{one} of those groups. Then, the effective size is defined as the number of such groups
\begin{equation}
N_{\mathrm{eff}}^{\mathrm{K}}(\psi) = n.\label{eq:29}
\end{equation}
A quantum state (\ref{eq:1}) is called macroscopic if $n = O(N)$.

In other words, we demand that the state (\ref{eq:1}) is \textit{locally distinguishable}. We pick a group $i$ of qubits such that we can distinguish between $\left| \psi_0 \right\rangle$ from $\left| \psi_1 \right\rangle $. Then, the ranges of the reduced density operators\footnote{Let $\rho$ be an $N$ qubit state and  $X$ be a subset of some qubits. The expression $\mathrm{Tr}_{X}\rho$ denotes the tracing over the qubits of $X$; $\mathrm{Tr}_{N\setminus X}\rho$ means that we trace over all qubits \textit{but} those of the set $X$.}  $\rho^i_{k} = \mathrm{Tr}_{N\setminus i} \left| \psi_k \rangle\!\langle \psi_k\right| $ ($k=0,1$) should be distinct, that is, $1/2 \lVert \rho_{0}^i - \rho_{1}^k\rVert_1 \geq 1-2\delta$.

This notion of macroscopic superposition is the only one in this paper that depends on a parameter: $\delta$ is the maximal failure probability measuring one block. The specific choice of $\delta$ influences in general the exact value of $N_{\mathrm{eff}}^{\mathrm{K}}$. For instance, the generalised GHZ state of example \ref{ex:generalized-ghz} gives $N_{\mathrm{eff}}^{\mathrm{K}}(\phi_{\epsilon})=N \frac{\log(\cos^2\epsilon)}{\log[4\delta(1-\delta)]} \approx N \epsilon^2 /[-\log(\delta)]$ for small $\epsilon$ and $\delta$. However, in this example, the choice of $\delta$ does not affect much the effective size, since it gives only a logarithmic factor. Even if we demand $\delta = O(N^{-1})$, the effective size would be $N_{\mathrm{eff}}^{\mathrm{K}}(\phi_{\epsilon})  \approx N \epsilon^2 /\log(N)$. On the other hand, consider
\begin{Exp}
  [PS + Domain Wall]\label{ex:ps-domainwall}
  We equally superpose the quantum states $\left| \psi_0 \right\rangle = \left| 0 \right\rangle ^{\otimes N}$ and
  \begin{equation}
    \left| \psi_1 \right\rangle = \frac{1}{\sqrt{N+1}} \sum_{k=0}^{N}\left| 1 \right\rangle^{\otimes k} \otimes \left| 0 \right\rangle ^{\otimes N-k}. 
    \label{eq:43}
\end{equation}
The state $\left| \psi_1 \right\rangle $ is called domain wall state. If we want to distinguish between $\left| \psi_0 \right\rangle $ and $\left| \psi_1 \right\rangle $ and measure a qubit $k$ by means of $\sigma_z$, we do not improve the success probability $P$ if we enlarge the group by another qubit $j>k$. The best strategy to group the qubits is to calculate $P$ for the first qubit, which is gives the highest $P$, and go on with the second, for which $P$ is reduced, and so on. For an arbitrary qubit $k$, one has $P = 1-\frac{k}{2(N+1)}$. We stop if for a certain qubit $j+1$, $P$ drops below the threshold $1-\delta$. So $j$ is the number of groups we find such that measuring any group let us distinguish between $\left| \psi_0 \right\rangle $ and $\left| \psi_1 \right\rangle $ with a probability of at least $P= 1 -\delta$. The effective size is therefore $N_{\mathrm{eff}}^{\mathrm{K}} = 2\delta (N+1)$.
\end{Exp} 

While it is obvious that $\delta$ should be small, that is, $\delta \ll 1$, the actual scaling one should demand is not entirely clear. In example \ref{ex:ps-domainwall}, for any choice $\delta = O(1)$ one would call $\left| \psi_0 \right\rangle + \left| \psi_1 \right\rangle $ a macroscopic state, whereas $\delta=1/N$ leads to a microscopic qualification.

We discuss this ambiguity of this measure by speculating about the idea behind this definition. We start with the original gedankenexperiment of Sch\"odinger's cat \cite{Schroedinger35} and Leggett's demand to call a superposition macroscopic if its constituents are macroscopically distinct \cite{Leg80}. On a macroscopic level, the distinction between two states $\left| \psi_0 \right\rangle $ and $\left| \psi_1 \right\rangle $ may be limited to sums of local measurements. We demand that the probability $P(\mathrm{error})$ for failing in distinguishing $\left| \psi_0 \right\rangle $ from $\left| \psi_1 \right\rangle $ may go to zero in the limit of large $N$, for instance, $P(\mathrm{error}) = 1/O(N)$. In \cite{DiscriminationBound}, it was shown that $P(\mathrm{error}) \leq \frac{1}{\Delta^2}$ with 
\begin{equation}
\label{eq:44}
\Delta = \frac{\left| \langle A \rangle_{\psi_0} - \langle A \rangle_{\psi_1} \right|}{\sqrt{\mathcal{V}_{\psi_0}(A)}+\sqrt{\mathcal{V}_{\psi_1}(A)}}.
\end{equation}
Observe that the quantity $\Delta^2$ is formally very similar to the effective size of Bj\"ork and Mana, equation (\ref{eq:41}), except that there we considered a time evolution under a Hamiltonian $H$ and here, in equation (\ref{eq:44}), the operator $A$ refers to an observable. In addition, for equation (\ref{eq:41}) we argued that we should maximise the numerator and the denominator separately in order to avoid divergences of $N_{\mathrm{eff}}^{\mathrm{B}}$. Here, we just need $\Delta^2 \geq O(N)$, in order to meet the requirement of a good distinguishability of  $\left| \psi_0 \right\rangle $ and $\left| \psi_1 \right\rangle $ with respect to $A$.

Now, the local distinctness between $\left| \psi_0 \right\rangle $ and $\left| \psi_1 \right\rangle $ as demanded by Korsbakken \textit{et al.}~leads for all choices of $\delta < 1/2$ to $\left| \langle A \rangle_{\psi_0} - \langle A \rangle_{\psi_1} \right|=O(N)$. To have $\Delta^2 \geq O(N)$, we need $\mathcal{V}_{\psi_i}(A) \leq O(N)$, $i=0,1$. This can be achieved by two different ways. Either we have no long-range quantum correlations between the local groups we are measuring on, that is,for almost all pairs of groups $(k,j)$ we have $\langle A^{(k)}A^{(j)} \rangle_{\psi_i} - \langle A^{(k)} \rangle_{\psi_i} \langle A^{(j)} \rangle_{\psi_i} = o(1)$. Then a finite $\delta=O(1)$ is enough to guarantee $P(\mathrm{error}) = 1/O(N)$. Or there exist $O(1)$ for $O(N^2)$ pairs of groups. Then the variances $\mathcal{V}_{\psi_i}(A)$ become large and we need $\delta = 1/O(N)$.

This insight is now applied to example~\ref{ex:ps-domainwall}. Since $\left| \psi_1 \right\rangle $ exhibits a large variance under the optimal observable, $\mathcal{V}_{\psi_i}(A) \approx N^2/12$, we would need $\delta = 1/O(N)$ and therefore, with this reasoning, we would conclude that this superposition is microscopic.

\subsection{Orthogonality with respect to local rotations}
\label{sec:orth-with-resp}
In a similar spirit as Korsbakken \textit{et al.}, the authors of \cite{MAvD}, Marquardt and coworkers, try to catch the idea of ``macroscopic distinctness''. They agree with preliminary works that the mere particle number of a quantum state is not enough to call the state macroscopic. Two orthogonal states of macroscopic size could differ only in a microscopic detail and their superposition should hence not be counted to the set of macroscopic superpositions. The actual ansatz they follow is to count the number of one particle operators one has to apply to $\left| \psi_0 \right\rangle$ \textit{in average} in order to reach a unity overlap with $\left| \psi_1 \right\rangle$. The mathematical definition of the effective size is described by these steps:

(1) Start with the one-dimensional subspace $\mathcal{H}_0$, spanned by $\left| \psi_0 \right\rangle$. Then all subspaces $\mathcal{H}_i$ with $i>0$ are generated iteratively.

(2) Take $\mathcal{H}_{i-1}$ and apply all possible one-particle operations on this subspace. This results in a set whose linear span is called $\tilde{\mathcal{H}}_i$.

(3) From this, define  $\mathcal{H}_i = \tilde{\mathcal{H}}_i \setminus \bigoplus_{k = 0}^{i-1} \mathcal{H}_k$. We assume that there is finite number of subspaces $\mathcal{H}_i$ such that the whole Hilbert space $\mathcal{H}$ is covered. To every $\mathcal{H}_i$ we define an orthogonal projection $P_i$. One has $\sum_iP_i = \mathbbm{1}$.

(4) Decompose the second state $\left| \psi_1 \right\rangle$ with respect to this division, that is
\begin{equation}
  \left| \psi_1 \right\rangle  = \sum_i \nu_i \left| \phi_i \right\rangle 
  \label{eq:27}
\end{equation}
with $P_i \left| \psi_1 \right\rangle = \nu_i \left| \phi_i \right\rangle$ such that $\left| \phi_i \right\rangle$ has unity norm.

(5) Define the effective size as the expectation value of the operator
\begin{equation}
  \label{eq:28}
  \mathcal{N} = \sum_kk \left| \phi_k \rangle\!\langle \phi_k \right| 
\end{equation}
under the second state, so one has
\begin{equation}
  N_{\mathrm{eff}}^{\mathrm{M}}(\psi) = \langle \mathcal{N} \rangle_{\psi_1}.\label{eq:25}
\end{equation}
The state $\left| \psi \right\rangle = \frac{1}{\sqrt{2}} \left( \left| \psi_0  \right\rangle +\left| \psi_1 \right\rangle \right)$ is called macroscopic, if $N_{\mathrm{eff}}^{\mathrm{M}}(\psi) = O(N)$.

As already noticed by the authors, there is a ambiguity in definition (\ref{eq:25}) of macroscopic states. The problem is that there exist superpositions $\left| \psi_0 \right\rangle +\left| \psi_1 \right\rangle $, where the decomposition (\ref{eq:27}) of $\left| \psi_1 \right\rangle$ with respect to the Hilbert space structure based on $\left| \psi_0 \right\rangle$ gives a different $N_{\mathrm{eff}}^{\mathrm{M}}$ than if we calculate the ``local distance'' of $\left| \psi_0 \right\rangle$ with respect to $\left| \psi_1 \right\rangle$. In spin systems, this can always happen if the entanglement structure of $\left| \psi_0 \right\rangle$ and $\left| \psi_1 \right\rangle$ differ. This is because with local projections, which can be written as a sum of identity and a Pauli operator, a highly entangled state can be quickly transformed into a product state. We illustrate this by the state of example~\ref{ex:ps-domainwall}. To reach $\left| \psi_1 \right\rangle $ with one particle operations from $\left| \psi_0 \right\rangle $ we need in average $\langle \mathcal{N} \rangle_{\psi_1}=\frac{N+1}{2}$ operations, since $\left| \psi_1 \right\rangle =1/\sqrt{N+1} \sum_{i=0}^N \left| \phi_i \right\rangle $. In contrast, we find $\left| \psi_0 \right\rangle = \frac{1}{2}\sqrt{N+1} \left| 0 \rangle\!\langle 0\right| \otimes \mathbbm{1}^{\otimes N-1} \left| \psi_1 \right\rangle = \frac{1}{2}\sqrt{N+1}\left( \mathbbm{1}^{\otimes N}+\sigma_z^{(1)} \right)\left| \psi_1 \right\rangle $, that is~$\langle \mathcal{N} \rangle_{\psi_0}<1$. In cases like example~\ref{ex:ps-domainwall}, the measure is not defined. Whenever the global entanglement structure of $\left| \psi_0 \right\rangle $ and $\left| \psi_1 \right\rangle $ are identical, we find at least the same order of $\langle \mathcal{N} \rangle_{\psi_0}$ and $\langle \mathcal{N} \rangle_{\psi_1}$. For the comparison in section \ref{sec:relat-among-meas}, we assume therefore that there exist local unitary operations $U_i$ acting on $O(1)$ qubits such that  $\left| \psi_1 \right\rangle = U_1\otimes \ldots \otimes U_n \left| \psi_0 \right\rangle $.
 
\section{Relations among the measures}
\label{sec:relat-among-meas}

In this section, we study the relations among macroscopic quantum states due to the respective measures. We find that for qubit systems we can establish a hierarchy among these measures. An important consequence is that measures \cite{BM,KWDC,MAvD} and the relative Fisher information, which are defined for superpositions of two orthogonal states [\textit{cf.} equation (\ref{eq:1})], detect a strict subclass of general macroscopic quantum states as detected by the index p \cite{p-index} and the Fisher information (\textit{cf.} section \ref{sec:fish-inform-as}).

With the Fisher information and the relative Fisher information, we introduced two proposals that are well defined for all quantum states and superposition states, respectively. In contrast, we encountered some ambiguities for the measures of Bj\"ork and Mana, Korsbakken \textit{et al.}~and of Marquardt \textit{et al.}. Therefore, we have to make some assumptions. As suggested in section \ref{sec:relat-fish-inform}, we maximise the numerator and denominator of equation (\ref{eq:41}) over all $A \in \mathcal{A}$ separately. For the local distinguishability in section \ref{sec:local-dist}, we demand a constant $\delta=O(1)$ for the success probability and only short-range quantum correlations for the quantum states $\left| \psi_0 \right\rangle $ and $\left| \psi_1 \right\rangle $ with respect to the local observable $A$. To circumvent problems of section \ref{sec:orth-with-resp}, we assume that there exist $n=O(N)$ distinct groups such that $\left| \psi_1 \right\rangle = U_1\otimes \ldots \otimes U_n \left| \psi_0 \right\rangle $.

We now summarise these observations. The proposition~\ref{Pro:Relations} is illustrated in figure \ref{fig:Hierarchy}. The results are proven at the end of the section.

\begin{Pro}
  \label{Pro:Relations}
Let $\left| \psi \right\rangle $ be an $N$ qubit quantum state. For all measures that rely on a superposition structure of $\left| \psi \right\rangle $, we have $\left| \psi \right\rangle = 1/\sqrt{2(1 + \mathrm{Re}\langle \psi_0 | \psi_1\rangle) }\left( \left| \psi_0 \right\rangle +\left| \psi_1 \right\rangle  \right)$  with $\langle \psi_0 | \psi_1\rangle = o(1)$ (i.e., the states become orthogonal in the limit of large $N$). For all discussed measures, a quantum state is considered to be macroscopic, if $N_{\mathrm{eff}}(\psi) = O(N)$. Under the assumptions formulated in the beginning of this section, then one has
\begin{enumerate}
\item\label{item:2} $N_{\mathrm{eff}}^{\mathrm{B}}(\psi) = O(N) \Leftrightarrow N_{\mathrm{eff}}^{\mathrm{RF}}(\psi) = O(N)$, [\textit{cf.} equations (\ref{eq:46}) and (\ref{eq:41})],

\item\label{item:1} $N_{\mathrm{eff}}^{\mathrm{B}}(\psi) = O(N) \Rightarrow N_{\mathrm{eff}}^{\mathrm{K}}(\psi) = O(N)$, [\textit{cf.} equations (\ref{eq:41}) and (\ref{eq:29})],
\item\label{item:3} $ N_{\mathrm{eff}}^{\mathrm{K}}(\psi) = O(N) \Leftrightarrow N_{\mathrm{eff}}^{\mathrm{M}}(\psi) = O(N)$, [\textit{cf.} equations (\ref{eq:29}) and (\ref{eq:25})], and
\item\label{item:4} $ N_{\mathrm{eff}}^{\mathrm{K}}(\psi) = O(N) \Rightarrow N_{\mathrm{eff}}^{F}(\psi) = O(N)$, [\textit{cf.} equations (\ref{eq:29}) and (\ref{eq:19})]. 

  \end{enumerate}
  \end{Pro}

  We therefore identify three different classes of macroscopic quantum states.
  
  (1) Bj\"ork and Mana/ relative Fisher information: quantum states that are superpositions of two ``classical'' (i.e., microscopic) states.  This is the most restrictive class of macro-states and is closest to the original Schr\"odinger cat.

(2) Korsbakken \textit{et al.}/ Marquardt \textit{et al.}: macroscopic superposition of locally distinct quantum states. Apart from the relation between the two states $\left| \psi_0 \right\rangle $ and $\left| \psi_1 \right\rangle $ there are no restrictions on them. Hence, these states can show true quantum effects with respect to other criteria (see example \ref{Ex:quantum-classical} later).

  (3) Index p/ Fisher information: \textit{a priori} no restrictions on $\left| \psi \right\rangle $. Quantum states of this category may not have anything in common with the original Schr\"odinger cat gedanken experiment, except that they can show counter-intuitive, macroscopic quantum effects.

\begin{figure}[htbp]
\centerline{\includegraphics[width=.5\columnwidth]{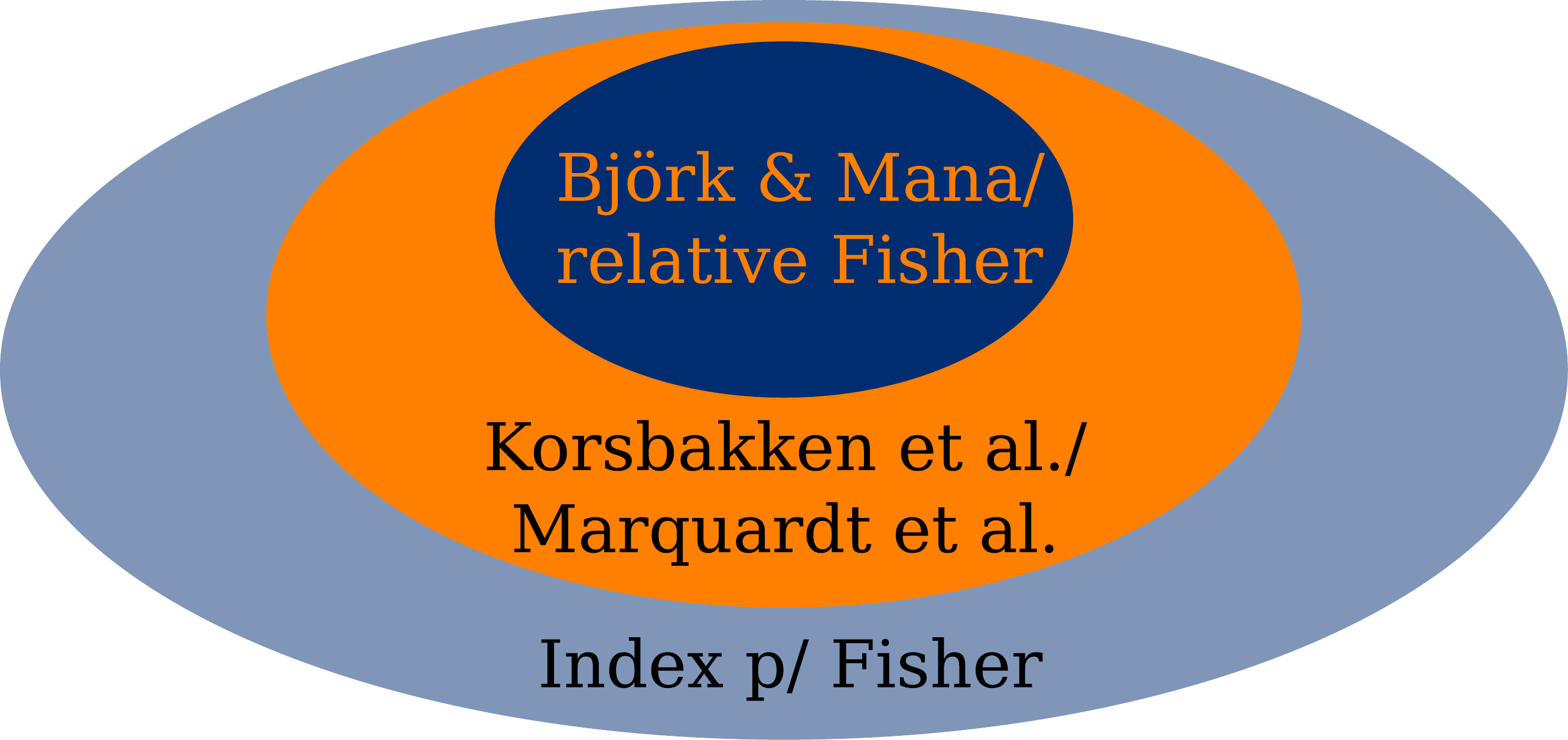}}
\caption[]{\label{fig:Hierarchy} Illustration of proposition~\ref{Pro:Relations}. The sets of macroscopic states detected by the respective measures are sketched by ellipses. One clearly identifies three different levels of macroscopic quantum states: superpositions of classical (Bj\"ork and Mana, relative Fisher information) and quantum states (Korsbakken \textit{et al.}~and Marquardt \textit{et al.}) and genuine macroscopic quantum states (Index p and Fisher information).} 
\end{figure}

The idea of the relative Fisher information introduced with definition \ref{def:relative-Fisher} is very similar to the proposal of Bj\"ork and Mana. And indeed, the connection between $N_{\mathrm{eff}}^{\mathrm{B}}$ and $N_{\mathrm{eff}}^{\mathrm{rF}}$ can be made more rigorous with part \ref{item:2} of proposition \ref{Pro:Relations}. It means that $\left| \psi \right\rangle $ is macroscopic due to Bj\"ork and Mana if and only if it is macroscopic due to the relative Fisher information. These two definitions therefore detect the same set of macroscopic quantum states.

Next, it is claimed in part \ref{item:1} that whenever a superposition is macroscopic due to Bj\"ork and Mana [and the relative Fisher information, see part \ref{item:2}], it is also macroscopic due to Korsbakken \textit{et al.}~[and Marquardt \textit{et al.}, see part \ref{item:3}]. The difference is that the first group demands that $\left| \psi_0 \right\rangle $ and $\left| \psi_1 \right\rangle $ are classical with respect to \textit{every} possible local measurement $A$. Here, by classical is meant that the standard deviation of $A$ for $\left| \psi_0 \right\rangle $ and $\left| \psi_1 \right\rangle $ is small (\textit{i.e.}, $O(\sqrt{N})$) compared to the spectral radius $\sigma$ of $A$. Therefore, in the limit of large $N$, the ratio $\sqrt{\mathcal{V}_{\psi_i}(A)}/\sigma(A) = 1/O(\sqrt{N})$ vanishes.

On the other side, the measures of Korsbakken \textit{et al.}~and Marquardt \textit{et al.}~are not so restrictive. There are local measurements for which the total state $\left| \psi_0 \right\rangle + \left| \psi_1 \right\rangle $ behaves macroscopically quantum but the single constituents do not. In contrast, there may be local observables for which also $\left| \psi_0 \right\rangle $ and $\left| \psi_1 \right\rangle $ show macroscopic behaviour, as in 
\begin{Exp}[Quantum/Classical]\label{Ex:quantum-classical}
  We divide the $N$ qubits into distinct groups of two qubits each. We relabel the basis states of one two-qubit block: $\left| \alpha_0 \right\rangle = \left| 00 \right\rangle $, $\left| \alpha_1 \right\rangle = \left| 11 \right\rangle $, $\left| \beta_0 \right\rangle = \left| 10 \right\rangle $ and $\left| \beta_1 \right\rangle = \left| 01 \right\rangle $. Consider two normalised quantum states
\begin{equation}
\label{eq:45}
\begin{split}
\left| \psi_0 \right\rangle &= \sum_{i_1,\ldots,i_{N/2} =0}^1 c_{i_1,\ldots,i_{N/2}} \left| \alpha_{i_1} \right\rangle \otimes \ldots \otimes \left| \alpha_{i_{N/2}} \right\rangle \\
  \left| \psi_1 \right\rangle &= \sum_{i_1,\ldots,i_{N/2}=0}^1 d_{i_1,\ldots,i_{N/2}} \left| \beta_{i_1} \right\rangle \otimes \ldots \otimes \left| \beta_{i_{N/2}} \right\rangle
\end{split}
\end{equation}
The superposition $\left| \psi_0 \right\rangle + \left| \psi_1 \right\rangle $ is for all choices of coefficients $c_{i_1,\ldots,i_{N/2}} \in \mathbbm{C}$ and $d_{i_1,\ldots,i_{N/2}}\in \mathbbm{C}$ macroscopic due to Korsbakken \textit{et al.}~and Marquardt \textit{et al}. For the first measure the observable $\sigma_z^{\otimes 2}$ always distinguishes perfectly between $\left| \psi_0 \right\rangle $ and $\left| \psi_1 \right\rangle $ on a local level. For the second measure, we need on every local group a unitary operation that maps the $\left\{ \alpha_i \right\}$ space to the orthogonal $\left\{ \beta_i \right\}$ space. Therefore, the distance between $\left| \psi_0 \right\rangle $ and $\left| \psi_1 \right\rangle $ in terms of local operations is in the order of $N$, i.e., the superposition is macroscopic.

In contrast, $\left| \psi_0 \right\rangle + \left| \psi_1 \right\rangle $ is in general not macroscopic due to the relative Fisher information, since the single constituents may show non-classical behaviour due to other local operators like $\sigma_x^{\otimes 2}$.
\end{Exp}

The third part of proposition \ref{Pro:Relations} states that a superposition is macroscopic due to Korsbakken \textit{et al.}~if and only if it is macroscopic due to Marquardt \textit{et al.}~under the conditions mentioned in the beginning of this section. It is clear that the absolute effective sizes can differ for both measures, which is also a consequence of the open parameter $\delta$ in the proposal of Korsbakken \textit{et al.}

Focusing on a superposition of two states $\left| \psi_0 \right\rangle $ and $\left| \psi_1 \right\rangle $ means that we implicitly assume that these states are a natural choice of quantum states with more or less classical behaviour. Dropping this restriction, we have to consider arbitrary macroscopic spin states, which leads to measures like the Fisher information or the index p. It is interesting that we find with proposition \ref{Pro:Relations} that all superposition states that are called macroscopic due to the measures of \cite{BM,KWDC,MAvD} and the relative Fisher information (\ref{eq:23}) are also macroscopic due to the more general measures \cite{p-index} and the Fisher information (definition \ref{DefExtension}). This indicates that --assuming that the last two proposals are well suited to classify macroscopic quantum states-- Schr\"odinger correctly caught the idea of macroscopic quantum phenomena with his cat-gedankenexperiment. But these phenomena can be demonstrated with much more general quantum states, which are not so illustrative any more. This means that with the help of local interactions within the system and with the environment to measure to systems, there exist other quantum states that do not rely on a superposition of two semi-classical states and still can show macroscopic quantum effects.

A counter-example that the reverse statement of part \ref{item:4} is not true is discussed extensively in \cite{ClonedCat}. We summarise it in 
\begin{Exp}
[Cloned superposition state]
  \label{Ex:Dicke}
  The eigenstates of $\sum_{i=1}^N \sigma_z^{(i)}$ with eigenvalue $N-2x$ that are completely symmetric under particle exchange are called Dicke-$x$ states $\left| N,x \right\rangle $ \cite{Dicke}. For odd $N$, we consider the superposition state $\left| \psi \right\rangle = \frac{1}{\sqrt{2}}\left( \left| \psi_0 \right\rangle  + \left| \psi_1 \right\rangle \right)$ with
\begin{equation}
\label{eq:51}
\begin{split}
  \left| \psi_0 \right\rangle &= \frac{1}{\sqrt{2}} \left( \left| N,\frac{N-1}{2} \right\rangle + \left| N,\frac{N+1}{2} \right\rangle  \right),\\
  \left| \psi_1 \right\rangle &= \frac{1}{\sqrt{2}} \left( \left| N,\frac{N-1}{2} \right\rangle - \left| N,\frac{N+1}{2} \right\rangle  \right).
\end{split}
\end{equation}
This superposition arises if an initial state $\left| \phi \right\rangle = \frac{1}{\sqrt{2}}\left( \left| 0 \right\rangle  + \left| 1 \right\rangle \right)$ is subject to an optimal phase covariant cloning device \cite{OptCloning}.

One can show that $\left| \psi \right\rangle $ has a variance that is macroscopic: $\mathcal{V}_{ \psi  }\left( \sum_{i=1}^N\sigma_x^{(i)} \right) = O(N^2)$. This state is therefore macroscopic due to the Fisher information and Shimizu and Miyadera (index p). On the contrary, the states (\ref{eq:51}) are not locally distinguishable and the superposition is therefore not macroscopic due to Korsbakken \textit{et al.} 
\end{Exp}
We now come to
\begin{proof}[Proof of proposition~\ref{Pro:Relations}]~\ref{item:2}
  We first derive the order of the denominator and the numerator of equation (\ref{eq:41}), if $N_{\mathrm{eff}}^{\mathrm{B}}(\psi)=O(N)$.
  The numerator and the denominator are individually optimised over all local Hamiltonians of the set $\mathcal{A}$ from definition~\ref{DefExtension} (b).
For convenience, the optimal operator is always denoted by $A$, even if it differs for different states. In general, $\sqrt{\mathcal{V}_{\psi_0}(A)} + \sqrt{\mathcal{V}_{\psi_1}(A)} \geq O(\sqrt{N})$, since one can always find a local operator for which the variance is at least the sum of the variances of the addends $A^{(i)}$. Furthermore, we have  $\langle A \rangle_{\psi_0}-\langle A \rangle_{\psi_1} \leq O(N)$ due to the spectral radius $\sigma(A)= N$.
  To fulfil both requirements simultaneously, one has $\sqrt{\mathcal{V}_{\psi_0}(A)} + \sqrt{\mathcal{V}_{\psi_1}(A)} = O(\sqrt{N})$ and $\langle A \rangle_{\psi_0}-\langle A \rangle_{\psi_1} = O(N)$. The same reasoning can be applied for the relative Fisher information. From $N_{\mathrm{eff}}^{\mathrm{rF}}(\psi)=O(N)$, it follows that $\mathcal{F}(\psi,A)=O(N^2)$, $\mathcal{F}(\psi_0,A)=O(N)$ and $\mathcal{F}(\psi_1,A)=O(N)$.

  We now come to the relation between $N_{\mathrm{eff}}^{\mathrm{B}}(\psi)$ and $N_{\mathrm{eff}}^{\mathrm{rF}}(\psi)$. The numerators of $N_{\mathrm{eff}}^{\mathrm{B}}(\psi)$ and $N_{\mathrm{eff}}^{\mathrm{rF}}(\psi)$ are always of the same order and we only have to look at the denominators. Let us consider $\left| \psi_i \right\rangle $ ($i=0,1$) in the eigenbasis $\left\{ \left| a_k \right\rangle  \right\}$ of $A$, $\left| \psi_i \right\rangle = \sum_k c_{ik} \left| a_k \right\rangle$. The conditions $\langle A \rangle_{\psi_0}-\langle A \rangle_{\psi_1} = O(N)$ and $\mathcal{V}_{\psi_i}(A) = O(N)$ imply that the state $\left| \psi \right\rangle $ exhibits at least two regions in the spectrum $\sigma(A)$ with non-vanishing probabilities to measure the corresponding eigenvalues. The variance of $\left| \psi \right\rangle $ reads $\mathcal{V}_{\psi}(A) = \sum_k \left| c_{0k} + c_{1k} \right|^2 (a_k - \langle A \rangle_{\psi})^2$ ($a_k$ denotes the eigenvalues of $A$). We see that $\mathcal{V}_{\psi}(A)$ contains at least two addends $k$ and $j$ with $a_k -a_j= O(N)$ and $\left| c_{0k} + c_{1k} \right|^2 = O(1)$ and $\left| c_{0j} + c_{1j} \right|^2 = O(1)$, which leads to $\mathcal{V}_{\psi}(A)= O(N^2)$. On the other side, $\mathcal{V}_{\psi}(A)= O(N^2)$ with $\mathcal{V}_{\psi_i}(A)= O(N)$ is only possible if $\langle A \rangle_{\psi_0}-\langle A \rangle_{\psi_1} = O(N)$, which establishes the equivalence of both measures in detecting macroscopic superpositions.

\ref{item:1} As we have seen in the first part of the proof, it follows from $N_{\mathrm{eff}}^{\mathrm{B}}(\psi) = O(N)$ that there exists a local operator $A \in \mathcal{A}$ such that $\left| \langle A  \rangle_{\psi_0}- \langle A \rangle_{\psi_1} \right| = O(N)$ and $\sqrt{\mathcal{V}_{\psi_0}(A)} + \sqrt{\mathcal{V}_{\psi_1}(A)} = O(\sqrt{N})$. Without loss of generality, we assume a locality of $A$ with respect to a fixed grouping into $n=O(N)$ blocks, each consisting of $O(1)$ qubits, i.e., $A=\sum_{i=1}^n A^{(i)}$. To fulfil $\left| \langle A  \rangle_{\psi_0}- \langle A \rangle_{\psi_1} \right| = O(N)$ there exist $O(N)$ blocks with $\lVert \rho_0^{i} - \rho_1^{i} \rVert_1 \geq \left| \langle A^{(i)}  \rangle_{\psi_0}- \langle A^{(i)} \rangle_{\psi_1} \right| = O(1)$. Therefore, measuring these blocks leads to a success probability $P = 1 -\epsilon$ with $\epsilon<1/2$. If $P > 1-\delta$, i.e., the demanded precision is fulfilled, then the actual grouping into $n$ blocks is the effective size of Korsbakken \textit{et al.}, $N_{\mathrm{eff}}^{\mathrm{K}} = n = O(N)$. If not, then we have to measure on a collection of $m=O(1)$ groups to improve the precision. Since the standard deviation of $\sum_{i = 1}^m A^{(i)}$ scales at most with $O(\sqrt{m})$, the success probability can be improved to $P\approx 1-\epsilon/\sqrt{m}$, which is for some $m$ larger than $1-\delta$. Then, the effective size reads $N_{\mathrm{eff}}^{\mathrm{K}} = n/m = O(N)$.

The implication $N_{\mathrm{eff}}^{\mathrm{K}}  =O(N) \Rightarrow N_{\mathrm{eff}}^{\mathrm{B}} = O(N)$ is in general not true (see example \ref{Ex:quantum-classical}).

\ref{item:3}If $N_{\mathrm{eff}}^{\mathrm{K}} = O(N)$, then we find $n=O(N)$ distinct groups of $O(1)$ qubits, such that for every group $i$ the ranges of the reduced density operators $\rho_k^i = \mathrm{Tr}_{N\setminus i} \left| \psi_k \rangle\!\langle \psi_k\right| $ ($k=0,1$) are almost orthogonal, i.e., $\frac{1}{2}\lVert \rho_0^i-\rho_1^i \rVert_1 = 1- 2\delta$ with $\delta \ll 1$. Therefore, one has to apply to every group $i$ a nontrivial rotation. In average, this corresponds to $O(1)$ one-particle operations. Since this has to be done for all groups simultaneously, we need in total $O(N)$ one-particle operations to transform $\left| \psi_0 \right\rangle $ to $\left| \psi_1 \right\rangle $ and hence we have $N_{\mathrm{eff}}^{\mathrm{M}} = O(N)$.

  On the other hand, $N_{\mathrm{eff}}^{\mathrm{M}} = O(N)$ requires that we have to rotate non-trivially on $O(N)$ groups of $\left| \psi_0 \right\rangle $ to have a large overlap with $\left| \psi_1 \right\rangle $. This means that there are $O(N)$ groups where the ranges of $\rho_0^i$ and $\rho_1^i$ are almost orthogonal. This makes the two states locally distinguishable, hence $N_{\mathrm{eff}}^{\mathrm{K}} = O(N)$.

\ref{item:4}  Assuming $N_{\mathrm{eff}}^{\mathrm{K}} = O(N)$, there exists a $A = \sum_{i=1}^n A^{(i)}\in \mathcal{A}$ such that $\left| \langle A \rangle_{\psi_0} - \langle A \rangle_{\psi_1} \right| = O(N)$ and $\mathcal{V}(\psi_i)(A) = O(N)$. As in part \ref{item:2}, we conclude that $\mathcal{V}(\psi_i)(A) = O(N^2)$, i.e., $N_{\mathrm{eff}}^{\mathrm{F}}(\psi) = N_{\mathrm{eff}}^{\mathrm{S}}(\psi) = O(N)$.

\end{proof}

\begin{table}[htb]
\begin{tabular}{l   c @{\quad} c@{\quad} c @{\quad} c}  \hline  \hline
  & $ N_{\mathrm{eff}}^{\mathrm{rF}} $
  & $N_{\mathrm{eff}}^{\mathrm{K}}$ 
  & $N_{\mathrm{eff}}^{\mathrm{M}}$ 
  & $N_{\mathrm{eff}}^{\mathrm{F}}$
  \\
  \hline
  GHZ (Ex.~\ref{ex:ghz})& $N$& $N$& $N$& $N$\\
  gen. GHZ (Ex.~\ref{ex:generalized-ghz})& $\approx N\sin^2
\epsilon + \cos^2 \epsilon$& $\approx N \frac{ \log(\cos\epsilon)}{-\log(\delta)}$& $N \sin^2 \epsilon$& $ =N_{\mathrm{eff}}^{\mathrm{rF}} $ \\ 
Cluster-GHZ (Ex.~\ref{ex:cluster-ghz})& $N/3$& $N/3$& $N/3$& $N/3$\\
\hline
PS+Domain Wall (Ex.~\ref{ex:ps-domainwall})& $O(1)$& ? & ? & $O(N)$\\
Quantum/Classical (Ex.~\ref{Ex:quantum-classical})& $O(1)$& $O(N)$& $O(N)$& $O(N)$\\
\hline  
Cloned superposition (Ex.~\ref{Ex:Dicke})& $O(1)$& $O(1)$& $O(1)$& $O(N)$\\
\hline  \hline\end{tabular}
\caption[]{\label{tab:summaryExamples} Overview on the examples of this paper for the discussed measures [from equations (\ref{eq:23}), (\ref{eq:29}), (\ref{eq:25}) and (\ref{eq:19}), respectively; note that $N_{\mathrm{eff}}^{\mathrm{B}} \approx N_{\mathrm{eff}}^{\mathrm{rF}}$ and $N_{\mathrm{eff}}^{\mathrm{S}}=N_{\mathrm{eff}}^{\mathrm{F}}$]. The examples are divided into three blocks. The first block are examples where all measures qualify the quantum states as macroscopic (with different absolute values). The second block of examples are microscopic due to Bj\"ork and Mana and the relative Fisher information, but macroscopic (or indefinite, see sections \ref{sec:local-dist} and \ref{sec:orth-with-resp}) for the other measures. The last example (see also \cite{ClonedCat}) is a macroscopic quantum state but not a macroscopic superposition. These three blocks correspond to the structure illustrated in figure \ref{fig:Hierarchy}.}
\end{table}

\section{Conclusion}
\label{sec:conclusion}

To summarise, we have introduced the concept of the quantum Fisher information as a measure for macroscopic quantum states in multipartite qubit systems. This definition has been motivated by the search for quantum states that potentially show the validity of quantum mechanics on a macroscopic scale under ``realistic'' conditions, which resulted in the restriction to local Hamiltonians for the time evolution and local measurements. We have seen that a high quantum Fisher information allows a rapid oscillation pattern for certain observables which cannot be reproduced by classical states. Furthermore, this kind of macroscopic quantum states are able to serve as a resource for high precision measurements.

This measure is well defined for generally mixed states and does not require a special shape of the state, for example, a superposition of two orthogonal quantum states. However, a slight modification allows us to judge on the ``catness'' of macroscopic superpositions. This variation was subsequently called relative Fisher information. With this proposal, we circumvent some ambiguities that other measures for macroscopic superpositions show.

We have contrasted our proposals to other measures for macroscopic quantum states suitable for qubit systems. We have suggested to distinguish carefully between ``macroscopic quantum states'' and ``macroscopic superpositions''. We have found that among those measures, a hierarchy can be established that classifies three different kinds of macroscopic quantum states: superpositions of orthogonal, semi-classical states; superpositions of \textit{a priori} general quantum states that are locally orthogonal; and general quantum states that exhibits long-range quantum correlations. Interestingly, it has turned out that macroscopic superpositions are a strict subset of general macroscopic quantum states.

This paper has focused on qubit systems. Although not presented explicitly, it is clear that one can generalise our findings to arbitrary spin systems. In future work, it will be interesting to investigate to role of the quantum Fisher information for macroscopicity in other frameworks like continuous variable systems.

\ack 
The authors thank Oriol Romero-Isart, Anika Pflanzer and Johannes Kofler for useful discussions and Tomoyuki Morimae for comments on the manuscript and advising us of references \cite{Morimae09} and \cite{Morimae10}. The research was funded by the Austrian Science Fund (FWF): P20748-N16, P24273-N16, SFB F40-FoQus F4012-N16 and the European Union (NAMEQUAM).


\section*{References}
\label{sec:references}

\bibliographystyle{iopart-num}
\bibliography{MasseBiblio}
\end{document}